\documentclass[preprint,12pt]{elsarticle}

\usepackage{graphicx}
\usepackage{amsmath,amssymb,amsfonts}
\usepackage{bm}
\usepackage[title]{appendix}%
\usepackage{xcolor}%
\usepackage{textcomp}%
\usepackage{manyfoot}%
\usepackage{booktabs}%
\usepackage{algorithm}%
\usepackage{algorithmicx}%
\usepackage{algpseudocode}%
\usepackage{listings}%
\usepackage{bbding}
\usepackage{amsthm}
\newtheorem{theorem}{Theorem}
\newtheorem{remark}{Remark}
\newtheorem{lemma}{Lemma}
\journal{Theoretical Computer Science}

\begin{document}

\begin{frontmatter}



\title{Constrained Distributed Heterogeneous Two-Facility Location Problems with Max-Variant Cost}

 \author{Xinru Xu\fnref{label1}}
\ead{xuxinru1207@stu.ouc.edu.cn}
 \author{Wenjing Liu\corref{cor1}\fnref{label1,label2}}
 \ead{liuwj@ouc.edu.cn}
 \author{Qizhi Fang\fnref{label1}}
 \ead{qfang@ouc.edu.cn}
 \author{Alexandros A. Voudouris\fnref{label3}}
 \ead{alexandros.voudouris@essex.ac.uk}

 \affiliation[label1]{
            addressline={School of Mathematical Sciences, Ocean University of China},
           city={Qingdao},
           postcode={266100},
         country={China}}

         \affiliation[label2]{
            addressline={Laboratory of Marine Mathematics, Ocean University of China},
           city={Qingdao},
           postcode={266100},
         country={China}}

 \affiliation[label3]{
            addressline={School of Computer Science and Electronic Engineering, University of Essex},
           city={Colchester},
           postcode={CO4 3SQ},
         country={United Kingdom}}

\cortext[cor1]{Corresponding author.}


 

%\affiliation{organization={},%Department and Organization
 %           postcode={}, 

%% Abstract
\begin{abstract}
This paper investigates a constrained distributed heterogeneous two-facility location problem under the max-variant cost model. In this setting, a set of agents with private locations on the real line is partitioned into disjoint groups. The constraint stipulates that facilities must be situated within a given multiset of candidate locations, with the restriction that each candidate location can host at most one facility. Under the max-variant model, an agent’s individual cost is defined as the distance from their location to the farthest facility. Our objective is to design strategyproof distributed mechanisms that incentivize agents to report their locations truthfully while approximating social objectives. Such mechanisms operate in two stages: first, for each group, a pair of candidate locations is selected as representatives based solely on local reports; subsequently, the mechanism outputs two final facility locations from the set of all representatives. We focus on a class of deterministic strategyproof distributed mechanisms and establish constant lower and upper bounds on the distortion under four social objectives: Average-of-Average, Max-of-Max, Average-of-Max, and Max-of-Average costs.
\end{abstract}



%% Keywords
\begin{keyword}
Facility location \sep
Mechanism design \sep  Strategyproof \sep Distributed \sep Distortion




\end{keyword}

\end{frontmatter}


\section{Introduction}
As a fundamental problem in combinatorial optimization, the facility location problem seeks to determine optimal facility placements under specific constraints to optimize a given social objective. In many practical settings, however, the social planner does not have direct access to agents' private information, such as their exact residential addresses. Instead, the planner must rely on information reported by the agents themselves. This creates a conflict of interest: while the planner aims for social optimality, individual agents may misreport their locations to reduce their own costs. The goal, therefore, is to design strategyproof mechanisms that ensure truthful reporting while (approximately) optimizing the social objective. This line of research on approximate mechanism design without money was  initiated by Procaccia and Tennenholtz~\cite{procaccia2013approximate} and has since led to a variety of models (see the survey by Chan et al.~\cite{chan2021mechanism}).

 In real-world scenarios, collective decision-making is often \textit{distributed}. In such processes, agents are divided into groups where each group first reaches a local decision independently of others; these local outcomes are then aggregated into a final collective decision. A typical example is the selection of national scholarship recipients at a university: each faculty first nominates its own candidates, and the final winners are chosen from this pool.  To analyze these complex dynamics, Filos-Ratsikas et al.~\cite{filos2020distortion} introduced the study of distributed social choice from a voting perspective, employing two-step mechanisms: each group elects a representative based on local preferences, and the mechanism then selects an overall winner from these representatives. To quantify the resulting efficiency loss, they extended the notion of \textit{distortion}, defined as the worst-case ratio between the social objective value achieved by a mechanism and the theoretical optimum. In follow-up work, Filos-Ratsikas and Voudouris~\cite{filos2021approximate} investigated a distributed single-facility location problem aimed at minimizing the social cost, which represents the sum of all agents' costs.

While classic models often assume that facilities can be placed anywhere in a metric space, many practical applications involve \textit{constrained} facility location. Feasible regions or the capacity of specific sites may be limited due to land-use regulations, environmental protection, or human factors. This has motivated research into mechanism design for facility location within restricted sets of locations.
%Feldman et al. \cite{feldman2016voting} studied the single-facility location setting in the context of voting embedded in a metric space, where the facility can only be built at a given set of candidate points of the line. 

 In this paper, we study a constrained distributed heterogeneous two-facility location problem. We assume each agent requires access to both facilities, and her individual cost is defined by the distance to the farther one, known as \textit{the max-variant cost}. This cost model reflects scenarios such as an express delivery outlet that must dispatch regular packages to an ordinary distribution center and cold-chain packages to a specialized one. Assuming the outlet has multiple vehicles traveling at the same speed, its total waiting time depends on the distance to the more distant center. Our work provides lower and upper  bounds on the distortion of strategyproof distributed mechanisms under four distinct social objectives.

\subsection{Our Results}
We study a constrained, distributed, and heterogeneous two-facility location problem. In this model, facilities can only be deployed within a predetermined set of candidate locations, with a capacity of at most one facility per site. We consider a set of agents with private locations on the real line, partitioned into disjoint groups. Once the two facilities are located, each agent’s individual cost is determined by the distance to the farther facility (the max-variant cost). Our objective is to design strategyproof distributed mechanisms that take agents' reported locations as input and output the final locations of the two facilities. Such mechanisms operate in two steps: first, for each group, the mechanism selects two candidate locations as representatives based solely on local reports; second, it chooses the final two facility locations from the aggregated pool of representatives.
Our mechanisms are evaluated under four social objectives: 
\begin{itemize}
\item \textbf{Average-of-Average cost}: The arithmetic mean of the average individual costs across all groups.
\item \textbf{Max-of-Max cost}: The maximum individual cost among all agents.
\item \textbf{Max-of-Average cost}: The maximum of the average individual costs of all groups.
\item \textbf{Average-of-Max cost}: The average of the maximum individual costs within each group.
\end{itemize}

While the Average-of-Average cost and the Max-of-Max cost are adaptations of classic objectives from non-distributed settings, the Max-of-Average cost and the  Average-of-Max cost  are  fairness-inspired metrics specifically meaningful in distributed environments. For each of these objectives, we establish lower and upper bounds on the distortion of strategyproof distributed mechanisms. A summary of our main results is provided in Table 1.
\begin{table}[H]\centering
\caption{Lower and upper  bounds on the distortion}\label{tab1}
\begin{tabular}{|l|c|c|}
\hline
Social objective & Lower bound & Upper bound  \\
\hline
Average-of-Average cost &   $3$ [Theorem \ref{AoAlower}]& $7$ [Theorem \ref{AoAupper}]\\
Max-of-Max cost & $3$ [Theorem \ref{MoMlower}]&  $3$ [Theorem \ref{MoMupper}]\\
Max-of-Average cost& $\frac{3+\sqrt{17}}{2}$ [Theorem \ref{MoAlower}]&$2+\sqrt{5}$ [Theorem \ref{MoAupper}]\\
Average-of-Max cost& $3$ [Theorem \ref{AoMlower}]& $2+\sqrt{5}$ [Theorem \ref{AoMupper}]\\
\hline
\end{tabular}  
\end{table}

\subsection{Related work}
The framework of approximate mechanism design without money was established by Procaccia and Tennenholtz \cite{procaccia2013approximate}, who developed strategyproof mechanisms with constant approximation ratios for facility location problems on the real line. Their foundational study covered both single-facility and homogeneous two-facility models under social cost and maximum cost objectives. Since then, the field has expanded significantly: Alon et al. \cite{alon2010strategyproof} explored facility location on circles and graphs, while Tang et al. \cite{tang2020mechanism} considered  constrained models where facilities are restricted to a discrete set of candidate locations. Other variations include obnoxious facility games, where agents seek to maximize their distance from the facility \cite{cheng2013strategy}, and studies focused on envy-based objectives \cite{cai2016facility,ding2020facility}. Below, we highlight the two research threads most relevant to our work: heterogeneous and distributed facility location.
%For minimizing the sum cost, they designed a Two-Extremes deterministic mechanism which is $(n-2)$- approximation and then Fotakis and Tzamos \cite{fotakis2014power} proved a matching lower bound of $n-2$, thus settling the sum cost objective of this setting. Further, Tang et al. \cite{tang2020mechanism} considered  constrained (or discrete) homogeneous two-facility location problems where the facilities can only be located in a given set of candidate locations. We will highlight the papers of heterogeneous two-facility location and distributed facility location problems respectively, which are most related to our work.

\paragraph{Heterogeneous facility location}
This line of research considers models where facilities may provide different utilities to agents. Zou and Li \cite{zou2015facility} investigated dual preferences where a facility could be either desirable or obnoxious. Serafino and Ventre \cite{serafino2016heterogeneous} introduced an optional preference model where each agent approves either one facility or both.  In their setting, agents’ locations are public while their preferences remain private, and the individual cost is defined as the sum of distances to the facilities of interest (referred to as sum-variant cost).  Later, Chen et al. \cite{chen2020facility} extended this optional preference model to include max-variant and min-variant costs, with the results for the latter subsequently refined by Li et al. \cite{li2021strategyproof}. Anastasiadis and
Deligkas \cite{anastasiadis2018heterogeneous} analyzed heterogeneous $k$-facility location problems with min-variant. In scenarios where agents' locations are private but preferences are public, Zhao et al. \cite{zhao2023constrained} explored the optional preference model with max-variant cost, while Kanellopoulos et al. \cite{kanellopoulos2025truthful} focused on the sum-variant cost. Additionally, Kanellopoulos and Voudouris \cite{kanellopoulos2025constrained} addressed obnoxious two-facility location with optional preferences, and Deligkas et al. \cite{deligkas2023heterogeneous} examined heterogeneous facility location under limited resources. Finally, Fong et al. \cite{fong2018facility} proposed a fractional preference model where an agent's preference for a facility is represented by a value between 0 and 1, a setting further studied by Fang and Liu \cite{fang2025heterogeneous} in the context of limited resources.

\paragraph{Distributed facility location} 
Filos-Ratsikas et al.~\cite{filos2020distortion} initiated the study of distortion in distributed social choice, which was further extended by Anshelevich et al. \cite{anshelevich2022distortion} and Voudouris \cite{Voudouris2023tight} to a distributed metric social choice setting where voters and alternatives are represented as points in a metric space. Voudouris \cite{voudouris2025metric} investigated the metric distortion of obnoxious distributed voting. Closely related to our work,  Filos-Ratsikas and Voudouris \cite{filos2021approximate} analyzed a distributed single-facility location problem under the social cost objective, considering both the discrete setting (where facilities are restricted to a finite set of candidate locations) and the continuous setting (where facilities can be placed at any point on the real line). For the discrete case, they established a tight bound of 7 for strategyproof distributed mechanisms, while for the continuous case, they derived a tight bound of 3.
Further, Filos-Ratsikas et al. \cite{filos2024distortion} explored a continuous distributed single-facility location problem under four social objectives: the average cost, the max cost, the average-of-max cost, and the max-of-average cost. A summary of their results is provided in Table 2.
\begin{table}[H]\centering
\caption{Tight bounds on the  distortion in Filos-Ratsikas et al. \cite{filos2024distortion}.}\label{tab1}
\begin{tabular}{|l|c|c|}
\hline
Social objective &  Tight bounds \\
\hline
Average cost &  $3$ \cite{filos2024distortion} \\
Max cost &  $2$ \cite{filos2024distortion} \\
Max-of-Average cost&$1+\sqrt{2}$ \cite{filos2024distortion} \\
Average-of-Max cost& $1+\sqrt{2}$ \cite{filos2024distortion} \\
\hline
\end{tabular}
\end{table}
%设施选址 Based on the distributed setting of Filos-Ratsikas et al. \cite{filos2024distortion}, we studied the distributed heterogeneous two-facility location problems.

\section{Preliminaries}
Let $N=\left\{1,2,...,n\right\}$ be a set of agents. Each agent $i\in N$ has a private location $x_i\in \mathbb{R}$ and denote by $\mathbf{x}=(x_1,...,x_n)\in \mathbb{R}^n$ the location profile of all agents. The agents are partitioned into $k$ ($\ge 1$) disjoint groups, denoted by the set $D=\left\{1,...,k\right\}$. For each group $d\in D$, let $N_d\subseteq N$ be the set of agents belonging to group $d$, and let $n_d=|N_d|$ represent the number of agents in that group.

Let $\mathcal F=\{ F_1,F_2\}$ be the two heterogeneous facilities to be located. Define $A=\left\{a_1,...,a_m\right\}\in \mathbb{R}^m$ as a multiset of candidate locations, with the constraint that at most one facility can be placed at each location in $A$. Without loss of generality, assume $a_1\le a_2\le\cdots\le a_m$. An instance of the problem is denoted by the tuple $I=(N,\mathbf{x},D,A)$. 

\paragraph{Distributed Mechanism} 
A distributed mechanism $M$ is a function that  maps an instance $I$ to a facility location profile $\mathbf{w}$ through two steps, i.e. $M(I)=\mathbf{w}=(w_1,w_2)$, where $w_1\in A$ is the location of $F_1$ and $w_2\in A\setminus\left\{w_1\right\}$ is the location of $F_2$. Specifically, for a given instance $I$, the mechanism $M$ operates in two steps: 
\begin{itemize}
    \item[$\bullet$] Step $1$. For each group $d\in D$, $M$ selects two representative locations $y_{d(1)}\in A$, $y_{d(2)}\in A\setminus \left\{y_{d(1)}\right\}$  based solely on the locations reported by agents in $N_d$.
    \item[$\bullet$] Step $2$. $M$ outputs a final facility location profile $\mathbf{w}=(w_1,w_2)$, where $w_1, w_2\in \{y_{d{(1)}},y_{d{(2)}}\}_{d \in D}$ and $w_1\neq w_2$, chosen from the set of all group representatives..
\end{itemize}

\begin{remark}
Following the distributed framework in Filos-Ratsikas et al. \cite{filos2024distortion},
Properties (P1) and (P2) represent the inherent structural characteristics of the mechanism, reflecting the fundamental information-restricted nature of the two-stage decision process:
\begin{itemize}
    \item \textbf{P1 (Local Aggregation)}: The representative of each group is determined solely by the positions of agents within that group.
    \item \textbf{P2 (Information Restriction)}: The final facility location profile must be a function of the group representatives, as the mechanism only has access to this summarized information in the second stage.
\end{itemize} 
These properties are also essential for deriving our distortion lower bounds.
\end{remark}

For any two points $x,y\in \mathbb{R}$, let $\delta(x,y)=|x-y|$ be the distance between $x$ and $y$. In our model, we assume that all agents approve the two heterogeneous facilities. The \textit{individual cost} of each agent $i$ for a facility location profile $\mathbf{w}$ is the distance to the farthest facility, which we refer to as the \emph{max-variant cost}:
\[c(x_i,\mathbf{w})=\max\left\{\delta(x_i,w_1),\delta(x_i,w_2)\right\}.
\]

\paragraph{Social Objectives} 
We consider four standard cost-minimization social objectives, defined as follows:
\begin{itemize}
\item[(1)] The \textit{Average-of-Average cost} of a facility location profile $\mathbf{w}$ is the arithmetic mean of the average individual costs across all groups:  
\[  
AoA(\mathbf{w}|I) = \frac{1}{k}\sum_{d \in D}\left\{\frac{1}{n_d} \sum_{i \in N_d} c(x_i, \mathbf{w}) \right\}. 
\]  

\item[(2)] The \textit{Max-of-Max cost} of a facility location profile  $\mathbf{w}$ is the maximum individual cost among all agents:  
\[  
MoM(\mathbf{w}|I) = \max_{d \in D} \max_{i \in N_d} c(x_i, \mathbf{w}).
\]

\item[(3)] The \textit{Max-of-Average cost} of a facility location profile  $\mathbf{w}$ is the maximum of the average individual costs of all groups:  
\[  
MoA(\mathbf{w}|I) = \max_{d \in D} \left\{\frac{1}{n_d}\sum_{i \in N_d} c(x_i, \mathbf{w})\right\}.
\]  

\item[(4)] The \textit{Average-of-Max cost} of a facility location profile  $\mathbf{w}$ is the average of the maximum individual costs within each group:  
\[  
AoM(\mathbf{w}|I) = \frac{1}{k}\sum_{d \in D} \max_{i \in N_d} c(x_i, \mathbf{w}).  
\] 
\end{itemize}

\paragraph{Strategyproofness} 
A mechanism $M$ is \textit{strategyproof }if each agent can never benefit by misreporting her location, regardless of the locations reported by the other agents, i.e., for every $i\in N$, for every $\mathbf{x}\in \mathbb{R}^n$,  and for every $x_i'\in \mathbb{R}$, it must hold that

\[  
c(x_i, M(N, (x_i, \mathbf{x}_{-i}), D, A) \leq c(x_i, M(N, (x_i', \mathbf{x}_{-i}), D, A))  ,
\]  
where \( \mathbf{x}_{-i} = (x_1, \ldots, x_{i-1}, x_{i+1}, \ldots, x_n) \) is the location profile of  $N\setminus\{i\}$.  

\paragraph{Distortion} 
In the distributed setting,  due to lack of global information and the requirement for strategyproofness of a mechanism, the facility location profile chosen by a distributed mechanism may be  suboptimal. Here, we adopt the notion of \textit{distortion} to quantify the performance gap between a mechanism and the optimal mechanism. Formally, the \textit{distortion} of a  distributed mechanism $M$ is the worst-case  ratio between the social objective value obtained by $M$ and the optimal social objective value, taken over all possible instances:
\[  
\text{dist}(M) = \sup_{I} \frac{\text{cost}(M(I) \mid I)}{\text{cost}(OPT(I) \mid I)}, 
\] 
where \(OPT(I)\) denotes the optimal solution for instance \(I\) and cost $\in\{AoA, $ $MoM, MoA, AoM\}$ represents the social objective function. For notational simplicity, we abbreviate $\text{cost}(M(I)|I)$ as $\text{cost}(M|I)$.

 For each agent $i\in N$, let $t({i})$ denote the closest candidate location to $i$ in the set $A$, and \(s(i) \) denote the closest candidate location to $i$ in $A\setminus \{t(i)\}$\footnote{In case of a tie, the left location is always chosen.}.  We now introduce a class of distributed mechanisms, termed the $(\alpha,\beta)$-Quantile Mechanism (with $\alpha,\beta\in [0,1]$).
 \vspace{1em}
 
\noindent\textit{\textbf{$\bm{(\alpha,\beta)}$-Quantile Mechanism}}\footnote{For the boundary cases: when $\alpha=0$, the $\lceil\alpha\cdot n_d\rceil$-th leftmost agent in $N_d$ is defined as the leftmost agent in $N_d$; similarly, when $\beta=0$, the $\lceil\beta\cdot k\rceil$-th leftmost location is defined as the leftmost location in the corresponding set.}
\begin{itemize}
    \item[$\bullet$] Step 1. For each group $d\in D$, let $\alpha_d$ be the $\lceil\alpha\cdot n_d\rceil$-th leftmost agent in $N_d$.
  Define $y_{d(1)}=t({\alpha_d})$ and $y _{d(2)}=s({\alpha_d})$ as the two representative locations for group $d$.
  \item[$\bullet$] Step 2. Initialize $z_d=y_{d(1)}$ for all $d\in D$ and let $w_1$ be the $\lceil\beta\cdot k\rceil$-th leftmost location in the set $\{z_d\}_{d\in D}$. For each group $d\in D$, update $z_d$ to $y_{d(2)}$ if $y_{d(1)}=w_1$, then set $w_2$ to be the $\lceil\beta\cdot k\rceil$-th leftmost location in the updated set $\{z_d\}_{d\in D}$.
\end{itemize}
 
\begin{remark}\label{remark2}
The $(\alpha,\beta)$-Quantile Mechanism naturally satisfies the following key properties.  For every group $d$, $y_{d(1)}$, $y_{d(2)}$ are consecutive candidate locations, which implies that the pairs $(y_{d(1)},y_{d(2)})$ and $(y_{d'(1)},y_{d'(2)})$ are never "interlaced" \footnote{Here, "interlaced" refers to the scenario where the representative pairs of two distinct groups overlap in order, i.e., $y_{d(1)}<y_{d'(1)}<y_{d(2)}<y_{d'(2)}$ for $d\neq d'$.} for any two distinct groups $d$ and $d'$. Consequently, the output locations $w_1$, $w_2$ of the mechanism must correspond to the two consecutive representative locations of a single group $d$, i.e., $w_1=y_{d(1)}$, $w_2=y_{d(2)}$.
    
   %\noindent (2) Without loss of generality, when there are several positions whose distances to $i_d$ are the same, the leftmost one prioritized.
\end{remark}

We will prove that the $(\alpha,\beta)$-Quantile Mechanism is strategyproof.

\begin{theorem}
The $(\alpha,\beta)$-Quantile Mechanism is strategyproof.
\end{theorem}
\begin{proof}
Consider any instance $I$, and let $\mathbf{w}=(w_1,w_2)$ be the facility location profile output by the mechanism. By the property of the mechanism established in Remark \ref{remark2}, $w_1$ and $w_2$ must be two consecutive candidate locations that are the closest and second-closest to some agent $j$, i.e., $w_1=t(x_j)$, $w_2=s(x_j)$.
  
Let $i$ be any agent in group $d$, and assume without loss of generality (w.l.o.g.) that $x_i \leq x_{\alpha_d}$ (the $\lceil \alpha \cdot n_d \rceil$-th leftmost agent in $N_d$). We analyze all possible orderings of $x_i, x_j, x_{\alpha_d}$ to show that agent $i$ has no incentive to misreport its location.

\medskip
\noindent
 \textbf{Case 1}: $x_i\le x_{\alpha_d}\le x_j$. To alter the mechanism's output, agent $i$ must first become the $\lceil \alpha \cdot n_d \rceil$-th leftmost agent in $N_d$, which requires reporting $x'_i > x_{\alpha_d}$. If $x_j\ge x_i' > x_{\alpha_d}$, the mechanism's output remains unchanged, so agent $i$ has no incentive to misreport. If $x_i' > x_j$, The output becomes the two consecutive locations closest to some agent $k$ with $x_k \geq x_j$. This cannot decrease agent $i$'s cost, so $i$ still has no incentive to misreport.

\medskip
\noindent
\textbf{Case 2}: $x_i<x_j< x_{\alpha_d}$. To alter the output, agent $i$ would need to report $x'_i > x_{\alpha_d}$ to become the $\lceil \alpha \cdot n_d \rceil$-th leftmost agent in $N_d$. However, this manipulation leaves the mechanism's output unchanged, so $i$ has no incentive to misreport.

\medskip
\noindent
\textbf{Case 3}: $x_j<x_i\le  x_{\alpha_d}$. This case is symmetric to Case 1, and agent $i$ has no incentive to misreport its location.
\end{proof}

In the subsequent sections, we analyze the upper bounds on the distortion of the $(\alpha, \beta)$-Quantile Mechanism under four social objective functions, by tuning the parameters $\alpha$ and $\beta$.

\paragraph{Notation} Let $l, r,$ and $m$ denote the leftmost, rightmost, and median agent in $N$, respectively. Similarly, let $l_d, r_d,$ and $m_d$ denote the leftmost, rightmost, and median agent in $N_d$, respectively. For a facility location profile $\mathbf{w} = (w_1, w_2)$ and any agent $i \in N$, let $w(x_i)$ denote the location in $\{w_1, w_2\}$ that is farther from agent $i$.

\section{Average-of-Average cost}

In this section, we study the Average-of-Average cost, defined as the arithmetic mean of the average individual costs across all groups. For the lower bound, we prove that the distortion of any strategyproof mechanism is at least $3-\epsilon$ for any $\epsilon>0$. For the upper bound, we analyze the $\left(\frac{1}{2}, \frac{1}{2}\right)$-Quantile Mechanism and show that it achieves a distortion of $7$, which is best possible among all $(\alpha,\beta)$-Quantile Mechanisms.  

\begin{theorem}\label{AoAlower}
For the Average-of-Average cost, the distortion of any strategyproof mechanism is at least $3-\epsilon$, for any $\epsilon>0$.
\end{theorem}

\begin{proof}
We proceed by contradiction. Assume there exists a strategyproof mechanism $M$ with distortion strictly smaller than $3 - \epsilon$ for some $\epsilon > 0$, and let $\theta > 0$ be an infinitesimal. We derive a contradiction by analyzing the following instances, with the set of candidate locations $A = \{0, 0, 1, 1\}$.

\medskip
\noindent
\textbf{(Instance $I_1$)}
Consider an instance $I_1$ with a single group containing one agent located at $x^1_1 = 0$.
By strategyproofness, $M$ must output $(0,0)$ as the group's representative, i.e., $M(I_1) = (0,0)$; otherwise, the distortion of $M$ would be unbounded.

\medskip
\noindent
\textbf{(Instance $I_2$)}
Consider an instance $I_2$ with a single group containing one agent located at $x^2_1 = \frac{1}{2} - \theta$.
If $M(I_2) \neq (0,0)$, then $c(x^2_1, M(I_2)) = \frac{1}{2} + \theta$.
Since $c(x^2_1, M(I_1)) = \frac{1}{2} - \theta$, the agent could strictly reduce her cost by misreporting her location as $0$, violating strategyproofness.
Thus, $M(I_2) = (0,0)$.

\medskip
\noindent
\textbf{(Instance $I_3$)}
Consider an instance $I_3$ with a single group containing one agent located at $x^3_1 = 1$.
By the same argument as in $I_1$, we obtain $M(I_3) = (1,1)$.

\medskip
\noindent
\textbf{(Instance $I_4$)}
Consider an instance $I_4$ with a single group containing one agent located at $x^4_1 = \frac{1}{2} + \theta$.
By the same argument as in  $I_2$, we conclude that $M(I_4) = (1,1)$.

\medskip
\noindent
\textbf{(Instance $I_5$)} Consider an instance $I_5$ with two groups:

\begin{itemize}  
    \item In group 1, there is a single agent located at $\frac{1}{2} - \theta$.
    \item In group 2, there is a single agent located at $1$.
\end{itemize} 

By $M(I_2) = (0,0)$ and $M(I_3) = (1,1)$, $(P1)$ implies that the representatives of group 1 and group 2 in $I_5$ are $(0,0)$ and $(1,1)$, respectively.  Since $AoA(0,0)=\frac{3}{4}-\frac{\theta}{2}$, $AoA(0,1)=AoA(1,0)=\frac{3}{4}+\frac{\theta}{2}$ and $AoA(1,1)=\frac{1}{4}+\frac{\theta}{2}$, Thus, $M$ must output $(1,1)$ as the global facility location profile, i.e., $M(I_5) = (1,1)$. Otherwise, the distortion of $M$ is at least $\frac{\frac{3}{2}-\theta}{\frac{1}{2}+\theta}\ge3-\epsilon$, for $\theta\le\frac{\epsilon}{8-2\epsilon}$.

\medskip
\noindent
\textbf{(Instance $I_6$)} We now derive a contradiction by considering the following instance $I_6$ with two groups:
\begin{itemize}
    \item In group 1, there is a single agent located at $\frac{1}{2} + \theta$.
    \item In group 2, there is a single agent located at $0$.
\end{itemize}

By $M(I_4) = (1,1)$ and $M(I_1) = (0,0)$, $(P1)$ implies that the representatives of group 1 and group 2 in $I_6$ are $(1,1)$ and $(0,0)$, respectively.
Since the group representatives in $I_5$ and $I_6$ are identical, $(P2)$ guarantees that $M(I_6) = (1,1)$.
Since $AoA(0,0)=\frac{1}{4}+\frac{\theta}{2}$, $AoA(0,1)=AoA(1,0)=\frac{3}{4}+\frac{\theta}{2}$ and $AoA(1,1)=\frac{3}{4}-\frac{\theta}{2}$, the distortion of $M$ is at least $\frac{\frac{3}{2}-\theta}{\frac{1}{2}+\theta}\ge3-\epsilon$, for $\theta\le\frac{\epsilon}{8-2\epsilon}$. This directly contradicts the assumption that $M$ achieves a distortion strictly smaller than $3 - \epsilon$.
\end{proof}

We next present the $\left(\frac{1}{2}, \frac{1}{2}\right)$-Quantile Mechanism and prove that its distortion is at most $7$.

\vspace{1em}
\noindent \textbf{$\bm{\left(\frac{1}{2}, \frac{1}{2}\right)}$-Quantile Mechanism} 
\begin{itemize}
    \item[$\bullet$] Step 1. For each group $d\in D$ , set $y_{d(1)}=t({m_d})$, $y _{d(2)}=s({m_d})$.
    \item[$\bullet$] Step 2. Initialize $z_d=y_{d(1)}$ for all $d\in D$ and let $w_1$ be the median location in the set $\{z_d\}_{d\in D}$. For each group $d\in D$, update $z_d$ to $y_{d(2)}$ if $y_{d(1)}=w_1$, then set $w_2$ to be the median location in the updated set $\{z_d\}_{d\in D}$.
\end{itemize}

\begin{theorem}\label{AoAupper}
For the Average-of-Average cost, the distortion of the $\left(\frac{1}{2}, \frac{1}{2}\right)$-Quantile Mechanism is at most $7$.
\end{theorem}

To prove Theorem~\ref{AoAupper}, we will need the following lemma which argues that it suffices to focus on instances with groups of equal size. 

\begin{lemma}\label{lem:equal-size-reduction}
For every instance $I$, there exists an instance $I'$ in which all groups have the same number of agents such that:
\begin{enumerate}
    \item For every facility location profile $\mathbf{w}$,
    $AoA(\mathbf{w}|I') = AoA(\mathbf{w}|I)$.
    \item The $(\alpha,\beta)$-Quantile Mechanism yields the same result for $I'$ as for $I$.
    
\end{enumerate}
\end{lemma}

\begin{proof}
Given an instance $I$, let $L := \text{lcm}(n_d)_{d \in D}$ be the least common multiple of the group sizes, and 
$c_d := \frac{L}{n_d}$; clearly, $c_d$ is an integer. 
We construct an instance $I'$ with equal-size groups by replacing every agent in group $d$ of $I$ with $c_d$ identical copies located at the same point. Then each group $d'$ in $I'$ has size $n_{d'} = c_d n_d = L$.

For the first part of the statement, given a facility location profile $\mathbf{w}$, the average cost of any group $d$ in $I$ is
$$\frac{1}{n_d} \sum_{i \in d} \delta(x_i,w(x_i)).$$
For the corresponding group $d'$ in $I'$, since each agent is replicated $c_d$ times, the average becomes
$$
\frac{1}{c_d n_d} \sum_{i \in d} c_d \cdot \delta(x_i,w(x_i)) 
= \frac{1}{n_d} \sum_{i \in g} \delta(x_i,w(x_i)).
$$
Therefore, each group’s average cost is unchanged, and the average-of-average cost remains unchanged as well. 

For the second part, fix a group $d$ in $I$, and sort the positions of the agents therein in non-decreasing order:
$$
x^{(d)}_1 \le x^{(d)}_2 \le \cdots \le x^{(d)}_{n_d}.
$$
Let $j_d := \lceil \alpha n_d \rceil$ and $j_{d'} := \lceil \alpha L \rceil = \lceil \alpha c_d n_d \rceil$.
By construction, in $d'$, $x^{(d)}_i$ appears $c_d$ times consecutively, and thus the copies of $x^{(d)}_{j_d}$ occupy the positions
$c_d(j_d - 1) + 1, \dots, c_d j_d$.
Observe that 
$$j_{d'} = \lceil \alpha c_d n_d \rceil \le c_d \lceil \alpha n_d \rceil = c_d j_d,$$
and, since $j_d - 1 < \alpha n_d$,
$$c_d(j_d - 1) < \alpha c_d n_d \le j_{d'}.$$
Hence, $j_{d'} \in \{c_d(j_d - 1) + 1, \dots, c_d j_d\}$, and the mechanism chooses the same representatives for $d'$ in $I'$ as for $d$ in $I$. 
Since the representatives are identical, the second step of the mechanism considers the same multiset of locations, and thus the mechanism outputs the same facility location profile. 

\end{proof}

Given Lemma~\ref{lem:equal-size-reduction}, it suffices to bound the distortion of any $(\alpha,\beta)$-Quantile Mechanism on instances with equal-size groups. Using this, we prove that the distortion of the $\left(\frac{1}{2}, \frac{1}{2}\right)$-Quantile Mechanism is at most $7$. To simplify notation, because all groups have the same size, we replace the average-of-average objective with the total sum of all agents’ costs. 

\begin{proof}[Proof of Theorem~\ref{AoAupper}]

Given any instance \( I \), let \(\mathbf{w} = (w_1, w_2) \) be the facility location profile produced by the mechanism, and let \( \mathbf{o} = (o_1, o_2) \) be an optimal solution. By the property of the mechanism, there exists a group \( d^* \) such that \( w_1 = y_{d^*(1)} = t(m_{d^*}) \) and \( w_2 = y_{d^*(2)} = s(m_{d^*}) \). Since \( w_1 \) and \( w_2 \) are two consecutive candidate locations, it holds that \( o \notin (w_1, w_2) \) for every \( o \in \{o_1, o_2\} \). If \( \mathbf{w} = \mathbf{o} \), the distortion equals 1; hence we only need to consider the two exhaustive cases below.

\medskip

\noindent \textbf{Case 1:} There exists \( o \in \{o_1, o_2\} \) such that \( w_2 < o \).  
For any agent \( i \) with \( x_i \leq \frac{w_1 + w_2}{2} \), we have \( \delta(x_i, w_2) \leq \delta(x_i, o) \). Applying this and the triangle inequality to agents with \( x_i > \frac{w_1 + w_2}{2} \), we obtain  

\[
\begin{aligned}
AoA(\mathbf{w}) &= \sum_{i: x_i \leq \frac{w_1 + w_2}{2}} \delta(x_i, w_2) + \sum_{i: x_i > \frac{w_1 + w_2}{2}} \delta(x_i, w_1) \\
&\leq \sum_{i \in N} \delta(x_i, o) + \left| \left\{ i : x_i > \tfrac{w_1 + w_2}{2} \right\} \right| \cdot \delta(w_1, o) \\
&= \sum_{i \in N} \delta(x_i, o) + \left( n - \left| \left\{ i : x_i \leq \tfrac{w_1 + w_2}{2} \right\} \right| \right) \cdot \delta(w_1, o).
\end{aligned}
\]

Because \( w_2 < o \), any agent \( i \) with \( x_i \leq \frac{w_1 + w_2}{2} \) satisfies \( \delta(x_i, o) \geq \delta(w_1, o)/2 \). Therefore,

\[
AoA(\mathbf{o}) \geq \sum_{i \in N} \delta(x_i, o) \geq \sum_{i: x_i \leq \frac{w_1 + w_2}{2}} \delta(x_i, o) \geq \left| \left\{ i : x_i \leq \tfrac{w_1 + w_2}{2} \right\} \right| \cdot \frac{\delta(w_1, o)}{2}.
\]

Consequently,

\[
\frac{AoA(\mathbf{w})}{AoA(\mathbf{o})} \leq 1 + 2 \cdot \frac{n - \left| \left\{ i : x_i \leq \tfrac{w_1 + w_2}{2} \right\} \right|}{\left| \left\{ i : x_i \leq \tfrac{w_1 + w_2}{2} \right\} \right|} = \frac{2n}{\left| \left\{ i : x_i \leq \tfrac{w_1 + w_2}{2} \right\} \right|} - 1.
\]

The median agent \( m_{d^*} \) of group \( d^* \) (which determines the mechanism’s output) is closer to \( w_1 \) than to \( w_2 \); thus she must lie weakly to the left of \( (w_1 + w_2)/2 \). By the definition of the mechanism, at least \( k/2 \) groups have their median agent weakly to the left of \( m_{d^*} \), and in each such group, at least half of the agents lie weakly to the left of the group’s median. Hence,
\[
\left| \left\{ i : x_i \leq \tfrac{w_1 + w_2}{2} \right\} \right| \geq \frac{k}{2} \cdot \frac{n}{2k} = \frac{n}{4},
\]
which yields
\[
\frac{2n}{\left| \left\{ i : x_i \leq \tfrac{w_1 + w_2}{2} \right\} \right|} - 1 \leq \frac{2n}{n/4} - 1 = 7.
\]

\medskip

\noindent \textbf{Case 2:} There exists \( o \in \{o_1, o_2\} \) such that \( o < w_1 \).  
Then \( \frac{o + w_2}{2} < \frac{w_1 + w_2}{2} \). Observe the following:

\begin{itemize}
\item For any agent \( i \) with \( x_i \geq \frac{w_1 + w_2}{2} \), we have \( \delta(x_i, w_1) \leq \delta(x_i, o) \).
\item For any agent \( i \) with \( x_i \in \bigl[ \frac{o + w_2}{2}, \frac{w_1 + w_2}{2} \bigr) \), we have \( \delta(x_i, w_2) \leq \delta(x_i, o) \).
\end{itemize}

Using these facts and applying the triangle inequality to agents with \( x_i < \frac{o + w_2}{2} \) (who incur cost \( \delta(x_i, w_2) \)), we get

\[
\begin{aligned}
AoA(\mathbf{w}) &= \sum_{i: x_i < \frac{o+w_2}{2}} \delta(x_i, w_2) + \sum_{i: \frac{o+w_2}{2} \leq x_i < \frac{w_1 + w_2}{2}} \delta(x_i, w_2) + \sum_{i: x_i \geq \frac{w_1 + w_2}{2}} \delta(x_i, w_1) \\
&\leq \sum_{i \in N} \delta(x_i, o) + \left| \left\{ i : x_i < \tfrac{o+w_2}{2} \right\} \right| \cdot \delta(o, w_2) \\
&\leq \sum_{i \in N} \delta(x_i, o) + \left( n - \left| \left\{ i : x_i \geq \tfrac{o+w_2}{2} \right\} \right| \right) \cdot \delta(o, w_2).
\end{aligned}
\]

Since \( o < w_1 \), any agent \( i \) with \( x_i \geq \frac{o + w_2}{2} \) satisfies \( \delta(x_i, o) \geq \delta(o, w_2)/2 \). Thus,

\[
AoA(\mathbf{o}) \geq \sum_{i \in N} \delta(x_i, o) \geq \sum_{i: x_i \geq \frac{o+w_2}{2}} \delta(x_i, o) \geq \left| \left\{ i : x_i \geq \tfrac{o+w_2}{2} \right\} \right| \cdot \frac{\delta(o, w_2)}{2}.
\]

Therefore,

\[
\frac{AoA(\mathbf{w})}{AoA(\mathbf{o})} \leq 1 + 2 \cdot \frac{n - \left| \left\{ i : x_i \geq \tfrac{o+w_2}{2} \right\} \right|}{\left| \left\{ i : x_i \geq \tfrac{o+w_2}{2} \right\} \right|} = \frac{2n}{\left| \left\{ i : x_i \geq \tfrac{o+w_2}{2} \right\} \right|} - 1.
\]

The median agent \( m_{d^*} \) is closer to \( w_2 \) than to \( o \), so she must lie weakly to the right of \( \frac{o + w_2}{2} \). By an argument analogous to Case 1, we have $\left| \left\{ i : x_i \geq \tfrac{o+w_2}{2} \right\} \right| \geq  \frac{n}{4},$ and therefore
\[
\frac{2n}{\left| \left\{ i : x_i \geq \tfrac{o+w_2}{2} \right\} \right|} - 1 \leq 7.
\]

\medskip

\noindent This completes the proof.

%By definition, at least $k/2$ $y_1$-representatives are (weakly) to the left of $w_1$, and in their corresponding groups, at least half of the agents are closer to these representatives than $w_2$. 

%This is also true for any other group with $y_1$-representative that is equal to $w_1$. Let $Z = \{d \neq d^*: t(m_d) = w_1 \wedge s(m_d)\geq s(m_{d^*}) \}$ be the set of groups with replaced representatives that are weakly to the right of $s(m_{d^*})$; let $z=|Z|$ and note that $z \in \{0, \ldots, \lfloor k/2 \rfloor \}$. By the definition of the mechanism, since $w_2$ is the median representative among the remaining $y_1$-representatives and the replaced representatives, there are at least $\lfloor \frac{k}{2} \rfloor - z$ groups with $y_1$-representatives that are (weakly) to the right of $w_2$, and in those groups, at least half of the agents are closer to these $y_1$-representatives than to $w_2$; hence, these agents are closer to $w_2$ than to $o$. 
%Putting everything together, we have
%$$\left|i: x_i \geq \frac{o+w_2}{2}\right| \geq (z+1)\frac{n}{2k} + \left(\left\lfloor \frac{k}{2} \right\rfloor - z \right) \frac{n}{2k} = \left(\left\lfloor \frac{k}{2} \right\rfloor +1 \right) \frac{n}{2k} \geq \frac{n(k+1)}{4k}.$$
\end{proof}

We also present a matching lower bound of $7$ on the distortion of any $(\alpha,\beta)$-Quantile Mechanism, showing that $\alpha=\beta=1/2$ are the best parameter choices within this class. 

\begin{theorem}
For any $\alpha, \beta \in (0,1)$, the distortion of the $(\alpha,\beta)$-Quantile Mechanism is at least 
\begin{align*}
    \max\left\{ \frac{2}{\alpha\beta}-1, \frac{2}{(1-\alpha)(1-\beta)}-1 \right\} \geq 7.
\end{align*}
\end{theorem}

\begin{proof}
Let $\varepsilon > 0$ be an infinitesimal.
Consider an instance with candidate locations $\{0,1,1\}$. 
There are $k$ groups, each of size $n/k$, such that:
    \begin{itemize}
        \item In the first $\beta k$ groups, $\alpha n/k$ agents are located at $1/2-\varepsilon$ while the remaining $(1-\alpha)n/k$ are located at $1$. For each such group, the $\alpha n/k$-th agent (when agents are sorted from left to right) is at $1/2 - \varepsilon$, so the representative pair is $(0,1)$. 
        
        \item In the remaining $(1-\beta)k$ groups, all $n/k$ agents are located at $1$. Clearly, for each such group, the representative pair is $(1,1)$. 
    \end{itemize}
    Among the $y_{d(1)}$-representatives of the groups, $\beta k$ of them are at $0$ and the remaining $(1-\beta)k$ are at $1$. Therefore, the mechanism selects $0$ as the first location, which is then replaced by $1$, leading to the final solution $(0,1)$. As $\varepsilon$ tends to $0$, the sum-cost of solution $(0,1)$ becomes
    \begin{align*}
        \beta k \left( \frac{\alpha n}{k}\cdot \frac12 + \frac{(1-\alpha) n}{k}  \right) + (1-\beta)k \cdot \frac{n}{k}
        = \left( 1 - \frac{\alpha \beta}{2} \right) n.
    \end{align*}
    On the other hand, the sum-cost of solution $(1,1)$ is
    \begin{align*}
        \beta k \cdot \frac{\alpha n}{k}\cdot \frac12 = \frac{\alpha\beta}{2} n.
    \end{align*}
    Hence, the distortion is at least
    \begin{align*}
        \frac{2}{\alpha \beta} - 1.
    \end{align*}

Consider a symmetric instance with candidate locations $\{0,0,1\}$. 
There are $k$ groups, each of size $n/k$, such that:
\begin{itemize}
    \item In the first $(\beta-\varepsilon) k$ groups, all $n/k$ agents are located at $0$. Clearly, for each such group, the representative pair is $(0,0)$. 
    \item In the remaining $(1-\beta+\varepsilon) k$ groups, $(\alpha-\varepsilon)n/k$ agents are located at $0$ while the remaining $(1-\alpha+\varepsilon) n/k$ agents are located at $1/2+\varepsilon$. For each such group, since the $\alpha n/k$-th agent (when agents are sorted from left to right) is located at $1/2+\varepsilon$, the representative pair is $(1,0)$.  
\end{itemize}
 By a derivation similar to the above instance, the mechanism will output $(1,0)$. As $\varepsilon$ tends to $0$, the sum-cost of solution $(1,0)$ becomes
\begin{align*}
        \beta k \cdot \frac{n}{k} + (1-\beta) k \left( \frac{\alpha n}{k} + \frac{(1-\alpha) n}{k} \cdot \frac12  \right)
        = \left( 1 - \frac{(1-\alpha)(1-\beta)}{2} \right) n.
    \end{align*}
 On the other hand, the sum-cost of solution $(0,0)$ is
    \begin{align*}
        (1-\beta) k \cdot \frac{(1-\alpha) n}{k}\cdot \frac12 = \frac{(1-\alpha)(1-\beta)}{2} n.
    \end{align*}
    Hence, the distortion is at least
    \begin{align*}
        \frac{2}{(1-\alpha)(1-\beta)} - 1.
    \end{align*}

By the above instances, the distortion of $(\alpha,\beta)$-Quantile Mechanism is 
\begin{align*}
    \max\left\{ \frac{2}{\alpha\beta}-1, \frac{2}{(1-\alpha)(1-\beta)}-1 \right\}.
\end{align*}
The left term is non-increasing in both both $\alpha$ and $\beta$, while the right term is non-decreasing. Hence, the minimum of the maximum is attained when the two terms are equal, i.e., $\alpha \beta = (1-\alpha)(1-\beta)$, which simplifies to $\alpha+\beta=1$. Substituting $\beta = 1-\alpha$ into either expression gives $\frac{2}{\alpha(1-\alpha)}-1$, which is minimized at $\alpha=1/2$, yielding a lower bound of $7$. This completes the proof.
\end{proof}

\section{Max-of-Max cost}

In this section, we study the Max-of-Max cost, defined as the maximum individual cost among all agents. 
For the lower bound, we prove that the distortion of any strategyproof mechanism is at least $3-\epsilon$ for any $\epsilon>0$. For the upper bound, we analyze the $(1,1)$-Quantile Mechanism\footnote{In fact, whatever values $\alpha$ and $\beta$ take, the $(\alpha,\beta)$-Quantile Mechanism would achieve a distortion of at most $3$ under the max-of-max cost.}, which achieves a distortion of at most $3$.

\begin{theorem}\label{MoMlower}
For the Max-of-Max cost, the distortion of any strategyproof mechanism is at least $3-\epsilon$, for any $\epsilon>0$.
\end{theorem}

\begin{proof}
We proceed by contradiction. Assume that there exists a strategyproof mechanism $M$ with distortion strictly smaller than $3-\epsilon$ for some $\epsilon>0$, and let $\theta>0$ be sufficiently small. We derive a contradiction by constructing a sequence of instances with candidate location set $A = \{0,0,2,2,4,4\}$.

\medskip
\noindent
\textbf{(Instance $I_1$)}
Consider an instance $I_1$ consisting of a single group with two agents located at $-1$ and $1-\theta$. By computing the MoM cost for all possible facility location pairs in $A$, we have $OPT(I_1)=(0,0)$. Thus, $M(I_1)=(0,0)$; otherwise its distortion would be at least $3$, contradicting the assumption.

\medskip
\noindent
\textbf{(Instance $I_2$)}
Consider an instance $I_2$ with a single group containing two agents both at $1-\theta$. If $M$ does not output $(0,0)$ on $I_2$, then the agent's cost would be at least $1+\theta$. Since the agent can reduce her cost to $1-\theta$ by misreporting her location as $-1$, strategyproofness implies $M(I_2)=(0,0)$.

\medskip
\noindent
\textbf{(Instance $I_3$)}
Consider an instance $I_3$ with a single group containing two agents at $5$ and $3+\theta$. By symmetry with $I_1$, we have $OPT(I_3)=(4,4)$, so $M$ must output $(4,4)$ to avoid distortion at least $3$.

\medskip
\noindent
\textbf{(Instance $I_4$)}
Consider an instance $I_4$ with a single group containing two agents both at $3+\theta$. Using the same strategyproofness argument as in $I_2$, we obtain $M(I_4)=(4,4)$.

\medskip
\noindent
\textbf{(Instance $I_5$)}
Consider an instance $I_5$ with two groups: group 1 has two agents at $1-\theta$, and group 2 has two agents at $3+\theta$. By $M(I_2)=(0,0)$ and $M(I_4)=(4,4)$, (P1) implies that the group representatives in $I_5$ are $(0,0)$ and $(4,4)$. Therefore, $M$ can only output $(0,0)$, $(0,4)$, $(4,0)$, or $(4,4)$, all of which yield MoM cost $3+\theta$. However, $OPT(I_5)=(2,2)$ and $MoM(OPT|I_5)=1+\theta$. As a result, the distortion of $M$ is at least
\[
\frac{3+\theta}{1+\theta} \ge 3-\epsilon,
\]
for sufficiently small $\theta$, which is a contradiction.
\end{proof} 

We next present the $(1,1)$-Quantile Mechanism and prove that its distortion is at most 3.

\medskip
\noindent 
\textbf{$\bm{(1,1)}$-Quantile Mechanism}  
\begin{itemize}
    \item[$\bullet$] Step 1. For each group $d\in D$ , set $y_{d(1)}=t({r_d})$, $y _{d(2)}=s({r_d})$.
    \item[$\bullet$] Step 2. Initialize $z_d=y_{d(1)}$ for all $d\in D$ and let $w_1$ be the rightmost location in the set $\{z_d\}_{d\in D}$. For each group $d\in D$, update $z_d$ to $y_{d(2)}$ if $y_{d(1)}=w_1$, then set $w_2$ to be the rightmost location in the updated set $\{z_d\}_{d\in D}$.
\end{itemize} 
\begin{theorem}\label{MoMupper}
For the Max-of-Max cost, the distortion of the $(1,1)$-Quantile mechanism is at most 3. 
\end{theorem}

\begin{proof}
	Given any instance $I$, let $\mathbf{w} = (w_1,w_2)$ be the location profile chosen by the mechanism and $\mathbf{o}= (o_1,o_2)$ be an optimal solution. By the property of the mechanism, we have  $w_1 = t(r)$ and $w_2 = s(r)$. Let $i^*$ denote the agent  such that $MoM(\mathbf{w}) = \max\{\delta(x_{i^*}, w_1), \delta(x_{i^*}, w_2)\} = \delta(x_{i^*}, w(x_{i^*}))$. As $w_1=t(r)$ and $w_2=s(r)$ are the closest candidate locations to $r$ and using the triangle inequality, we have 
	\begin{equation*}
		\begin{aligned}
			MoM(\mathbf{w}) &= \delta(x_{i^*}, w(x_{i^*})) \\
			&\leq \delta(x_{i^*}, o_1) + \delta(o_1, x_r) + \delta(x_r, w(x_{i^*}))\\
			&\le \delta(x_{i^*}, o(x_{i^*})) + \delta(x_r, o(x_r)) + \delta(x_r, w(x_r)) \\
			&\le \delta(x_{i^*}, o(x_{i^*})) + 2 \cdot \delta(x_r, o(x_r)).   
		\end{aligned}
	\end{equation*}
	
For the optimal solution $\mathbf{o}$, we have 
	$MoM(\mathbf{o}) \geq \delta(x_i, o(x_i))$ for all $i\in N$.  
	Thus, we obtain
\[MoM(\mathbf{w})\leq \delta(x_{i^*},o(x_{i^*}))+2\cdot \delta(x_r,o(x_r))\leq3\cdot MoM(\mathbf{o}).\]
\end{proof}

\section{Max-of-Average cost}

In this section, we study the Max-of-Average cost, defined as the maximum of the average individual costs of all groups.
For the lower bound, we prove that the distortion of any strategyproof mechanism is at least $\frac{3+\sqrt{17}}{2}-\epsilon$, for any $\epsilon>0$.
For the upper bound, we analyze the  $(\frac{3-\sqrt{5}}{2},1)$-Quantile Mechanism, which achieves a distortion of at most $2+\sqrt{5}$. 

\begin{theorem}\label{MoAlower}
For the Max-of-Average cost, the distortion of any strategyproof mechanism is at least $\frac{3+\sqrt{17}}{2}-\epsilon$, for any $\epsilon>0$.
\end{theorem}

\begin{proof}
We proceed by contradiction. Assume that there exists a strategyproof mechanism $M$ with distortion strictly smaller than $\frac{3+\sqrt{17}}{2}-\epsilon$ for some $\epsilon>0$. Let $\lambda=\frac{\sqrt{17}-1}{4}$ and $\theta>0$ be sufficiently small. Consider the following instances with candidate location set $A = \{0,0,1,1,2,2\}$.

\medskip
\noindent 
{\bf (Instance $I_1$)}
Consider an instance $I_1$, consisting of a single group $d_1$ with $\lambda n$ agents at $0$ and $(1-\lambda)n$ agents at $1$. Then $OPT(I_1)=(0,0)$ with MoA cost $1-\lambda$.
The MoA cost of $(1,1)$ is $\lambda$, and the MoA cost of any other solution is at least $(\lambda  + (1-\lambda)n)/n = 1$. Thus, $M(I_1)=(0,0)$; otherwise, the distortion of $M$ would be at least 
\begin{align*}
    \frac{\lambda}{1-\lambda} = \frac{3+\sqrt{17}}{2}.
\end{align*}

\medskip
\noindent 
{\bf (Instance $I_2$)}
Starting from $I_1$, move the $\lambda n$ agents at $0$ one by one  to $1/2-\theta$ to obtain instance $I_2$. 
By strategyproofness, $M(I_2)=(0,0)$; otherwise, any such agent at $1/2-\theta$ would prefer to misreport her location as $0$ to reduce her cost from $1/2+\theta$ to $1/2-\theta$. 
Hence, we obtain a group $d_2$ with $\lambda n$ agents at $1/2-\theta$ and $(1-\lambda)n$ agents at $1$ that has representatives $(0,0)$. 

\medskip
\noindent 
{\bf (Instance $I_3$)}
Symmetrically, start with an instance consisting of a single group $d_1'$ with $\lambda n$ agents at $2$ and $(1-\lambda)n$ agents at $1$. By the same reasoning, the optimal solution is (2,2) and the mechanism must output (2,2) on that instance.  Then move the agents at $2$ one by one to $3/2+\theta$ to obtain an instance $I_3$. By strategyproofness, $M(I_3)=(2,2)$, yielding a group $d_2'$ with $\lambda n$ agents at $3/2+\theta$ and $(1-\lambda)n$ agents at $1$ that has representatives $(2,2)$. 

\medskip
\noindent 
{\bf (Instance $I_4$)}
Consider an instance $I_4$ with two groups: one copy of $d_2$ and one copy of $d_2'$, which have representatives $(0,0)$ and $(2,2)$, respectively. Hence, $M$ can only output $(0,0)$, $(0,2)$, or $(2,2)$. The MoA cost of any such solution is $\lambda \cdot (3/2 + \theta) + (1-\lambda)$. On the other hand, the MoA cost of solution $(1,1)$ is $\lambda \cdot (1/2 + \theta)$. Thus, the distortion of $M$ is at least
\begin{align*}
    \frac{\frac32 \lambda + (1-\lambda)+\lambda\theta}{\frac12 \lambda\theta } \ge \frac{3+\sqrt{17}}{2}-\epsilon,
\end{align*}
for sufficiently small $\theta$, which leads to a contradiction. 
\end{proof}

We first present the $(\alpha,1)$-Quantile Mechanism and analyze its performance.
\vspace{1em}

\noindent \textbf{$\bm{(\alpha,1)}$-Quantile Mechanism} 

\begin{itemize}
    \item[$\bullet$] Step 1. For each group $d\in D$ , let $\alpha_d$ denote the $\lceil\alpha\cdot n_d\rceil$-th leftmost agent in $N_d$.  Set $y_{d(1)}=t({\alpha_d})$, $y _{d(2)}=s({\alpha_d})$.
    \item[$\bullet$]  Step 2. Initialize $z_d=y_{d(1)}$ for all $d\in D$ and let $w_1$ be the rightmost location in $\{z_d\}_{d\in D}$. For each group $d\in D$, update $z_d$ to $y_{d(2)}$ if $y_{d(1)}=w_1$, then set $w_2$ to be the rightmost location in $\{z_d\}_{d\in D}$.
\end{itemize} 
 
\begin{lemma}
    For the Max-of-Average cost, the distortion of the $(\alpha,1)$-Quantile Mechanism is at most $\max \left\{ 1+\frac{2(1-\alpha)}{\alpha}, 1 + \frac{2}{1 - \alpha}\right\}$.
\end{lemma}

\begin{proof}
    Given any instance $I$, let \( \mathbf{w}= (w_1, w_2) \) be the location profile chosen by the mechanism and \( \mathbf{o} = (o_1, o_2) \) be an optimal solution. Assume w.l.o.g. that \( w_1 \leq w_2 \) and \( o_1 \leq o_2 \).  By the property of the mechanism, there exists a group $d^*$ such that \( w_1 = y_{d^*(1)} = t({\alpha_{d^*}}) \) and $w_2 = y_{d^*(2)} =s(\alpha_{d^*})$.
Let \( d' \) denote the group with the maximum average individual cost under \( \mathbf{w} \), so that $MoA(\mathbf{w})=\frac{1}{n_{d'}}\sum_{i\in N_{d'}}\delta(x_i,w(x_i) )$.  Let $N_{d'(1)}=\{i\in N_{d'}|w(x_i)=w_1\}$ and $N_{d'(2)}=\{i\in N_{d'}|w(x_i)=w_2\}$.

\medskip
\noindent
\textbf{Case 1}: \( o_1 \leq o_2 < w_1 \leq w_2 \). By the definition of \( d' \) and the triangle inequality, we have
\begin{align}
       MoA(\mathbf{w})
         &=\frac{1}{n_{d'}}\sum_{i \in N_{d'(1)} } \delta(x_i, w_1) + \frac{1}{n_{d'}}\sum_{i \in N_{d'(2)} } \delta(x_i, w_2) \nonumber\\
             &\leq\frac{1}{n_{d'}} \sum_{i \in N_{d'} } \delta(x_i, o(x_i)) + \frac{1}{n_{d'}}\sum_{i \in N_{d'(1)} } \delta(o_1, w_1) + \frac{1}{n_{d'}}\sum_{i \in N_{d'(2)} } \delta(o_2, w_2) \nonumber\\  
            &\leq MoA(\mathbf{o}) + \delta(o_1, w_2). \label{13}
\end{align}
Let $S=\left\{i\in N_{d^*}|x_i\ge x_{\alpha_{d^*}}\right\}$. By the definition of $\alpha_{d^*}$, $\left|S\right|\ge (1-\alpha)\cdot n_{d^*}$. Since \( o_1 \leq o_2 < w_1 \leq w_2 \) and $w_1$, $w_2$ are the closest  locations to $\alpha_{d^*}$, all agents in $S$ are closer to $\mathbf{w}=(w_1,w_2)$ than $\mathbf{o}=(o_1, o_2)$. So we have $\delta(x_i,o_2)\ge \frac{\delta(o_2,w_2)}{2}$ and $\delta(x_i,o_1)\ge \delta(o_1,o_2)+\frac{\delta(o_2,w_2)}{2}$ for all $i\in S$.
Using these properties, we have
\begin{equation*}
    \begin{aligned}
MoA(\mathbf{o})&\ge\frac{1}{n_{d^*}}\sum_{i \in N_{d^*} } \delta(x_i, o(x_i))\ge\frac{1}{n_{d^*}}\sum_{i \in S } \delta(x_i, o_1) \\
             &\ge(1-\alpha)\cdot\left[\frac{\delta(o_2,w_2)}{2}+\delta(o_1,o_2)\right],
\end{aligned}
\end{equation*}
and
\begin{equation*}
    \begin{aligned}
       MoA(\mathbf{o})\ge\frac{1}{n_{d^*}}\sum_{i \in S } \delta(x_i, o_2) 
            \ge(1-\alpha)\cdot\frac{\delta(o_2,w_2)}{2}.
\end{aligned}
\end{equation*}
Then, it follows that
\begin{equation}\label{16}
    \begin{aligned} 
MoA(\mathbf{o})&\ge\frac{(1-\alpha)}{2}\cdot\delta(o_1,w_2).
\end{aligned}
\end{equation}
Combining (\ref{13}) and (\ref{16}), we obtain
\begin{equation*}
    \begin{aligned} 
       MoA(\mathbf{w})&\le MoA(\mathbf{o})+\delta(o_1,w_2)\le \left(1+\frac{2}{1-\alpha}\right)\cdot MoA(\mathbf{o}).
\end{aligned}
\end{equation*}

\medskip
\noindent
\textbf{Case 2}: \( o_1 <o_2 = w_1 < w_2 \). Similar to Case 1, we obtain
\begin{equation*}
    \begin{aligned} 
       MoA(\mathbf{w})\le \left(1+\frac{2}{1-\alpha}\right)\cdot MoA (\mathbf{o}).
\end{aligned}
\end{equation*}

\medskip
\noindent
\textbf{Case 3}: \( w_1\le w_2 <o_1 \le  o_2 \). Let $L$ be the set of agents in $N_{d'}$ from the first leftmost agent to $\alpha_{d'}$ and $R$ be the set of remaining agents in $N_{d'}$.
 By the definition of $\alpha_{d'}$, we have $\left|L\right|=\alpha\cdot n_{d'}$ and $\left|R\right|=(1-\alpha)\cdot n_{d'}$. Since $w_1=\max\left\{y_{d(1)}\right\}_{d\in D}$, it follows that $y_{d'(1)}\le w_1\le w_2<o_1\le o_2$. As $y_{d'(1)}$ is the closest location to $\alpha_{d'}$, we have 
$\delta(x_i,w(x_i))\le \delta(x_i,o(x_i))$ for all $i\in L$. Using these properties, we obtain
\begin{align}
MoA(\mathbf{w}) &= \frac{1}{n_{d'}}\sum_{i \in L} \delta(x_i, w(x_i)) + \frac{1}{n_{d'}}\sum_{i \in R} \delta(x_i, w(x_i)) \nonumber\\ 
&\leq \frac{1}{n_{d'}}\sum_{i \in L} \delta(x_i, o(x_i)) +\frac{1}{n_{d'}} \sum_{i \in R \cap N_{d'(1)}} \delta(x_i,w_1)+ \frac{1}{n_{d'}}\sum_{i \in R \cap N_{d'(2)}} \delta(x_i,w_2) \nonumber\\
&\leq \frac{1}{n_{d'}}\sum_{i \in L} \delta(x_i, o(x_i)) + \frac{1}{n_{d'}}\sum_{i \in R \cap N_{d'(1)}} [\delta(x_i,o_1)+\delta(o_1,w_1)] \nonumber \\ 
&+ \frac{1}{n_{d'}}\sum_{i \in R \cap N_{d'(2)}} [\delta(x_i, o_2) + \delta(o_2, w_2)] \nonumber \\ 
&\leq MoA(\mathbf{o}) + (1 - \alpha) \cdot \delta(o_2, w_1). \label{19}
\end{align}
Since all agents in $L$ are closer to $\mathbf{w}=(w_1,w_2)$ than $\mathbf{o}=(o_1,o_2)$,  we have $\delta(x_i,o_2)\ge \frac{\delta(w_1,w_2)}{2}+\delta(o_2,w_2)$ (when $w_1\ne w_2$) and $\delta(x_i,o_2)\ge \frac{\delta(o_1,w_2)}{2}+\delta(o_1,o_2)$ (when $w_1=w_2$) for all $i\in L$. Thus,
\begin{align}
 MoA(\mathbf{o})&\ge\frac{1}{n_{d'}}\sum_{i \in N_{d'} } \delta(x_i, o(x_i))\ge\frac{1}{n_{d'}}\sum_{i \in L } \delta(x_i, o_2) \nonumber \\
             &\ge\alpha\cdot\left[\frac{\delta(w_1,w_2)+\delta(o_2,w_2)}{2} \right]\ge \frac{\alpha}{2}\cdot\delta(o_2,w_1). \label{20}
\end{align}
Combining (\ref{19}) and (\ref{20}), we obtain
\begin{equation*}
    \begin{aligned} 
       MoA(\mathbf{w})\le MoA(\mathbf{o})+(1-\alpha)\cdot\delta(o_2,w_1)\le \left(1+\frac{2(1-\alpha)}{\alpha}\right)\cdot MoA(\mathbf{o}).
\end{aligned}
\end{equation*}

\medskip
\noindent
\textbf{Case 4}: \(  w_1 <w_2=o_1\le o_2 \). Similar to Case 3, we obtain
\[MoA(\mathbf{w})\le \left(1+\frac{2(1-\alpha)}{\alpha}\right)\cdot MoA(\mathbf{o}).\]

\medskip
\noindent
\textbf{Case 5}: \(o_1< w_1\le w_2\le o_2 \). By the definition of $d'$ and the triangle inequality, we have
\begin{align}
MoA(\mathbf{w})&=\frac{1}{n_{d'}}\sum_{i \in N_{d'(1)} } \delta(x_i, w_1) + \frac{1}{n_{d'}}\sum_{i \in N_{d'(2)} } \delta(x_i, w_2)\nonumber\\
 &\leq MoA(\mathbf{o}) + \delta(o_1, o_2).\label{MoA9}
\end{align}
Clearly, $MoA(\mathbf{o})\ge \frac{1}{n_{d'}}\sum_{i \in N_{d'} } \delta(x_i, o_1)$ and $MoA(\mathbf{o})\ge\frac{1}{n_{d'}} \sum_{i \in N_{d'} } \delta(x_i, o_2)$. Adding these two inequalities and applying the triangle inequality again, we have
\begin{equation}\label{MoA10}
    \begin{aligned}
MoA(\mathbf{o})&\ge\frac{1}{2n_{d'}}\cdot\sum_{i \in N_{d'} } [\delta(x_i, o_1)+\delta(x_i,o_2)]\ge\frac{1}{2}\cdot\delta(o_1,o_2)
\end{aligned}
\end{equation}
Combining (\ref{MoA9}) and (\ref{MoA10}), we obtain $MoA(\mathbf{w})\le 3\cdot MoA(\mathbf{o}).$

\medskip
\noindent
\textbf{Case 6}: \(o_1=w_1\le w_2<o_2 \). Similar to Case 5, we obtain $MoA(\mathbf{w})\le 3\cdot MoA(\mathbf{o}).$

\medskip
\noindent
Combining all cases, an upper bound of the distortion is given by 
\[
\max\left\{ 1 + \frac{2(1-\alpha)}{\alpha},\ 1 + \frac{2}{1-\alpha},\ 3 \right\} = \max\left\{ 1 + \frac{2(1-\alpha)}{\alpha},\ 1 + \frac{2}{1-\alpha} \right\}.
\]
\end{proof}

By solving the equation $1 + \frac{2(1 - \alpha)}{\alpha} = 1 + \frac{2}{1 - \alpha}$, we establish the following theorem.

\begin{theorem}\label{MoAupper}
For the Max-of-Average cost, the distortion of the $(\frac{3 - \sqrt{5}}{2},1)$-Quantile Mechanism is at most \( 2 + \sqrt{5} \).  
\end{theorem}

\section{Average-of-Max cost}

In this section, we study the Average-of-Max cost, defined as the average of the maximum individual costs within each group. For the lower bound, we prove that the distortion of any strategyproof mechanism is at least $3-\epsilon$ for any $\epsilon>0$. For the upper bound, we analyze the $(1,\frac{3-\sqrt{5}}{2})$-Quantile Mechanism, which achieves a distortion of at most $2+\sqrt{5}$. 

\begin{theorem}\label{AoMlower}
For the Average-of-Max cost, the distortion of any strategyproof mechanism is at least $3-\epsilon$, for any $\epsilon>0$.
\end{theorem}

\begin{proof}
We proceed by contradiction. Assume there exists a strategyproof mechanism $M$ with distortion strictly smaller than $3 - \epsilon$ for some $\epsilon > 0$, and let $\theta > 0$ be an infinitesimal. We derive a contradiction by analyzing the following instances, with the set of candidate locations $A = \{0, 0, 1, 1\}$.

\medskip
\noindent
\textbf{(Instance $I_1$)}
Consider an instance $I_1$ with a single group containing two agents located at $\frac{1}{2}$ and $\frac{3}{2}$.
We compute the AoM costs for all candidate facility pairs and obtain that $OPT(I_1) = (1,1)$, so $M(I_1) = (1,1)$ must hold; otherwise, the distortion of $M$ would be at least $3$, contradicting our assumption.

\medskip
\noindent
\textbf{(Instance $I_2$)}
Consider an instance $I_2$ with a single group containing two agents, with $x^2_1 = \frac{1}{2}$ and $x^2_2 = \frac{1}{2} + \theta$.
If $M(I_2) \neq (1,1)$, then $c(x^2_2, M(I_2)) = \frac{1}{2} + \theta$.
Since $c(x^2_2, M(I_1)) = \frac{1}{2} - \theta$, the agent at $x^2_2$ could strictly reduce her cost by misreporting her location as $\frac{3}{2}$, violating strategyproofness.
Thus, $M(I_2) = (1,1)$.

\medskip
\noindent
\textbf{(Instance $I_3$)}
Consider an instance $I_3$ with a single group containing one agent, with $x^3_1 = 0$.
By strategyproofness, the mechanism $M$ must output $(0,0)$ as the group's representative, i.e., $M(I_3) = (0,0)$; otherwise, the distortion of $M$ would be infinite.

\medskip
\noindent
\textbf{(Instance $I_4$)}
Consider an instance $I_4$ with a single group containing one agent, with $x^4_1 = \frac{1}{2} - \theta$.
If $M(I_4) \neq (0,0)$, then $c(x^4_1, M(I_4)) = \frac{1}{2} + \theta$.
Since $c(x^4_1, M(I_3)) = \frac{1}{2} - \theta$, the agent at $x^4_1$ could strictly reduce her cost by misreporting her location as $0$, violating strategyproofness.
Thus, $M(I_4) = (0,0)$.

\medskip
\noindent
\textbf{(Instance $I_5$)}
Consider an instance $I_5$ with two groups: group 1 contains one agent at $\frac{1}{2} - \theta$, whose representative is $(0,0)$ by $I_4$, and group 2 contains one agent at $1$, whose representative is $(1,1)$ by the same argument as $I_3$.
We compute the AoM costs for all possible facility location pairs:
\[
AoM(0,0) = \frac{3}{4} - \frac{\theta}{2}, \quad AoM(0,1) = AoM(1,0) = \frac{3}{4} + \frac{\theta}{2}, \quad AoM(1,1) = \frac{1}{4} + \frac{\theta}{2}.
\]
Thus, $M(I_5) = (1,1)$.
Otherwise, the distortion of $M$ would be at least $\frac{\frac{3}{4} - \frac{\theta}{2}}{\frac{1}{2} + \theta} \geq 3 - \epsilon$ for $\theta \leq \frac{\epsilon}{8 - 2\epsilon}$, which contradicts our assumption.

\medskip
\noindent
\textbf{(Instance $I_6$)}
We now reach a contradiction by considering instance $I_6$ with two groups: group 1 contains two agents at $\frac{1}{2}$ and $\frac{1}{2} + \theta$, whose representative is $(1,1)$ by $I_2$, and group 2 contains one agent at $0$, whose representative is $(0,0)$ by $I_3$.
By $I_5$ and $(P2)$, we conclude that $M(I_6) = (1,1)$.
We compute the AoM costs for all facility location pairs:
\[
AoM(0,0) = \frac{1}{4} + \frac{\theta}{2}, \quad AoM(0,1) = AoM(1,0) = \frac{3}{4} + \frac{\theta}{2}, \quad AoM(1,1) = \frac{3}{4}.
\]
Thus, the distortion of $M$ is at least $\frac{\frac{3}{4}}{\frac{1}{2} + \theta} \geq 3 - \epsilon$ for $\theta \leq \frac{\epsilon}{8 - 2\epsilon}$, which is a contradiction.
\end{proof}

We next present the $(1,\beta)$-Quantile Mechanism and analyze its performance.

\medskip
\noindent 
\textbf{$\bm{(1,\beta)}$-Quantile Mechanism} 

\begin{itemize}
\item[$\bullet$] Step 1. For each group $d\in D$ , set $y_{d(1)}=t({r_d})$, $y _{d(2)}=s({r_d})$.

\item[$\bullet$] Step 2. Initialize $z_d=y_{d(1)}$ for all  $d\in D$ and let $w_1$ be the $\lceil\beta\cdot k\rceil$-th leftmost location in $\{z_d\}_{d\in D}$. For each group  $d\in D$, update $z_d$ to $y_{d(2)}$ if $y_{d(1)}=w_1$, then set $w_2$ to be the $\lceil\beta\cdot k\rceil$-th  leftmost location in $\{z_d\}_{d\in D}$.
\end{itemize} 

\begin{lemma}
    For the Average-of-Max cost, the distortion of the $(1,\beta)$-Quantile mechanism is at most $\max \left\{ 1+\frac{2(1-\beta)}{\beta}, 1 + \frac{2}{1 - \beta}\right\}$.
\end{lemma}

\begin{proof}
    Given any instance $I$, let \( \mathbf{w}= (w_1, w_2) \) be the solution chosen by the mechanism and \( \mathbf{o} = (o_1, o_2) \) be an optimal solution. W.l.o.g., we assume that \( w_1 \leq w_2 \) and \( o_1 \leq o_2 \).  By the property of the mechanism, there exists a group $d^*$ such that \( w_1 = y_{d^*(1)} = t(r_{d^*}) \) and $w_2 = y_{d^*(2)} = s(r_{d^*})$. 
For each group $d$, let $i_d$ and $i'_{d}$ denote the agents in $N_d$ with the maximum individual cost under $\mathbf{w}$ and $\mathbf{o}$, respectively. Then we have 
\begin{equation*}
    \begin{aligned}
       &AoM(\mathbf{w}) =\frac{1}{k} \sum_{d \in D} \left\{\max_{i \in N_d} \delta(x_i, w(x_i))\right\} = \frac{1}{k}\sum_{d \in D} \delta(x_{i_{d}}, w(x_{i_{d}})),\\
      & AoM(\mathbf{o}) = \frac{1}{k}\sum_{d \in D}\left\{ \max_{i \in N_d} \delta(x_i, o(x_i))\right\} = \frac{1}{k}\sum_{d \in D} \delta(x_{i_{d'}}, o(x_{i'_{d}})).
\end{aligned}
\end{equation*}

Let $D_1=\{d\in D|w(x_{i_d})=w_1\}$ and $D_2=\{d\in D|w(x_{i_d})=w_2\}$. 

\medskip
\noindent
\textbf{Case 1}: \( o_1 \leq o_2 < w_1 \leq w_2 \).  Using the triangle inequality, and since \( \delta(x_{i_d}, o(x_{i_{d}})) \leq \delta(x_{i'_{d}}, o(x_{i'_{d}})) \), we have  \begin{equation}\label{25}
    \begin{aligned}
       AoM(\mathbf{w})& = \frac{1}{k}\sum_{d \in D_1} \delta(x_{i_d}, w_1) +\frac{1}{k} \sum_{d \in D_2} \delta(x_{i_d}, w_2) \\
       &\leq \frac{1}{k}\sum_{d \in D} \delta(x_{i_d}, o(x_{i_d})) + \frac{|D_1|}{k} \cdot \delta(o_1, w_1) +\frac{|D_2|}{k}  \cdot \delta(o_2, w_2)\\
       &\leq AoM(\mathbf{o}) +  \delta(o_1, w_2).
\end{aligned}
\end{equation}
 Let $S=\{ d\in D\mid y_{d(1)}\ge w_1, y_{d(2)}\ge w_1\}$.
 For each group $d\in D$, since $y_{d(1)}$, $y_{d(2)}$ are consecutive, \( w_1 \) is the  $\lceil\beta\cdot k\rceil$-th leftmost location in \( \{ y_{d(1)} \}_{d \in D} \) and \( w_2 \) is the$\lceil\beta\cdot k\rceil$-th leftmost location in the  updated \( \{ z_{d} \}_{d \in D}\), we have  \( |S| \geq (1 - \beta) \cdot k \). Since \( w_1 \) and \( w_2 \) are the closest locations to \( x_{r_{d^*}} \), it holds that \( x_{r_{d^*}}\ge  \frac{o_2 + w_2}{2} \). Then, we have $\delta(x_{r_d}, o_1) \geq \delta(o_1, o_2) + \frac{\delta(o_2, w_2)}{2}$ for all \( d \in S \). Hence,
\begin{equation}\label{26}
    \begin{aligned}
       AoM(\mathbf{o})\ge \frac{1}{k}\sum_{d \in S} \delta(x_{r_d}, o_1)  \ge \frac{(1-\beta)}{2}\cdot\delta(o_1,w_2).
\end{aligned}
\end{equation}
Combining (\ref{25}) and (\ref{26}), we have 
\begin{equation*}
    \begin{aligned}
       AoM(\mathbf{w})\leq AoM(\mathbf{o}) +  \delta(o_1, w_2) \le \left(1+\frac{2}{1-\beta}\right) \cdot AoM(\mathbf{o}).   
\end{aligned}
\end{equation*}

\medskip
\noindent
\textbf{Case 2}: \( o_1 < o_2 = w_1 <w_2 \). Similar to Case 1, we obtain $AoM(\mathbf{w})\le \left(1+\frac{2}{1-\beta}\right)\cdot AoM(\mathbf{o})$.

\medskip
\noindent
\textbf{Case 3}: \( w_1 \le w_2 < o_1 \le o_2 \). Let \( L \) denote the set of groups from the leftmost group in \( \{ y_{d(1)} \}_{d \in D} \) to the group at the $\lceil\beta\cdot k\rceil$-th position in \( \{ y_{d(1)} \}_{d \in D} \) and \( R \) the set of  remaining  groups. Then, for every \( d \in L \), $i\in N_d$, since \( y_{d(1)} \) is the closest location to \( x_{r_d} \) and \( y_{d(1)} \leq w_1 \le w_2< o_1 \leq o_2 \),  we have $\delta(x_i, w(x_i)) = \delta(x_i,w_2)\le  \delta(x_i, o(x_i)) = \delta(x_i, o_2)$.
By the triangle inequality, we have  
\begin{align}
      AoM(\mathbf{w})& =\frac{1}{k} \sum_{d \in L} \delta(x_{i_d}, w_2) +\frac{1}{k} \sum_{d \in R} \delta(x_{i_d}, w(x_{i_d}))\nonumber \\
      &\le\frac{1}{k}\sum_{d \in L} \delta(x_{i_d}, o(x_{i_d})) +\frac{1}{k} \sum_{d \in R\cap D_1} \delta(x_{i_d}, w_1) +\frac{1}{k} \sum_{d \in R\cap D_2} \delta(x_{i_d}, w_2) \nonumber  \\
      &\le AoM(\mathbf{o})+(1-\beta) \cdot \delta(o_2, w_1).\label{28}
\end{align}
Since \( w_1 \) and \( w_2 \) are the closest  locations to \( x_{r_{d^*}} \), it holds that 
$x_{r_{d^*}} \leq \frac{w_1 + w_2}{2} \quad \text{(when \( w_1 \neq w_2 \))} \quad \text{or} \quad x_{r_{d^*}} \leq \frac{w_2 + o_1}{2} \quad \text{(when \( w_1 = w_2 \))}$.  Then, for every \( d \in L \), we have   $\delta(x_{r_d}, o_2) \geq  \frac{\delta(w_2, o_2) +\delta(w_1, w_2)}{2}.$
Hence,  
\begin{equation}\label{29}
    \begin{aligned}
     {AoM}(\mathbf{o}) \geq\frac{1}{k} \sum_{d \in L} \delta(x_{r_d}, o_2) \geq \frac{\beta }{2} \cdot \delta(o_2, w_1).
\end{aligned}
\end{equation}
Combining (\ref{28}) and (\ref{29}), we obtain  
\begin{equation*}
    \begin{aligned}
      {AoM}(\mathbf{w})\le AoM(\mathbf{o})+(1-\beta)\cdot\delta(o_2,w_1)\le \left(1+\frac{2(1-\beta)}{\beta}\right)\cdot AoM(\mathbf{o}).
\end{aligned}
\end{equation*}

\medskip
\noindent
\textbf{Case 4}: \( w_1 < w_2 = o_1 \le o_2 \).  Similar to Case 3, we obtain ${AoM}(\mathbf{w})\le \left(1+\frac{2(1-\beta)}{\beta}\right)\cdot AoM(\mathbf{o})$.

\medskip
\noindent
\textbf{Case 5}: \( o_1 < w_1 \le w_2 \le o_2 \).  Using the triangle inequality again, we have
\begin{equation*}
    \begin{aligned}
      AoM(\mathbf{w})& =\frac{1}{k} \sum_{d \in D} \delta(x_{i_d}, w(x_{i_d}))\leq AoM(\mathbf{o}) + \delta(o_1, o_2). 
\end{aligned}
\end{equation*}
Clearly, $AoM(\mathbf{o})\ge\frac{1}{k}\sum_{d \in D} \delta(x_{i'_{d}}, o_1)$ and $AoM(\mathbf{o})\ge\frac{1}{k}\sum_{d \in D} \delta(x_{i'_{d}}, o_2)$. Using again the triangle inequality, we have
\begin{equation*}
    \begin{aligned}
AoM(\mathbf{o})&\ge\frac{1}{2k}\cdot\sum_{d \in D } [\delta(x_{i'_{d}}, o_1)+\delta(x_{i'_{d}},o_2)]\ge\frac{1}{2}\cdot\delta(o_1,o_2).
\end{aligned}
\end{equation*}
Therefore, we conclude that $AoM(\mathbf{w})\le 3\cdot AoM(\mathbf{o}).$

\medskip
\noindent
\textbf{Case 6}: \(o_1=w_1\le w_2<o_2 \). Similar to Case 5, we obtain $AoM(\mathbf{w})\le 3\cdot AoM(\mathbf{o}).$

\medskip
\noindent
Combining all cases, an upper bound of the distortion is given by
\[\max \left\{ 1 + \frac{2(1 - \beta)}{\beta}, \, 1 + \frac{2}{1 - \beta},3 \right\}=\max \left\{ 1 + \frac{2(1 - \beta)}{\beta}, \, 1 + \frac{2}{1 - \beta} \right\}.\]  
\end{proof}

By solving the equation $1 + \frac{2(1 - \beta)}{\beta} = 1 + \frac{2}{1 - \beta}$, we establish the following theorem.

\begin{theorem}\label{AoMupper}
For the Average-of-Max cost, the distortion of the  $(1,\frac{3-\sqrt{5}}{2})$-Quantile Mechanism  is at most $2+\sqrt{5}$.
\end{theorem}

\section{Conclusions and Future Work}

In this paper, we studied a constrained distributed heterogeneous two-facility location problem, and established lower and upper bounds on the distortion of strategyproof mechanisms under four social objectives.
There are several promising directions for future work. First, it would be of great interest to close the gaps between our derived lower and upper bounds for each objective. Second, it is meaningful to extend our analysis to settings with more than two facilities, as well as to more general metric spaces beyond the real line. Third, one could explore alternative social objectives in distributed environments (such as minimax envy), and consider scenarios where agents have more general preferences over facilities (e.g., fractional preferences).

\medskip
\noindent  
\textbf{Acknowledgments.}
   This research was supported in part by the National Natural Science Foundation of China (12201590, 12171444) and Natural Science Foundation of Shandong Province (ZR2024MA031).

\medskip   
\noindent
\textbf{Disclosure of Interests.}
\begin{footnotesize}
   The authors declare that they have no conflict of interest.
\end{footnotesize}

\end{document}